\documentclass[11pt,onecolumn,twoside]{IEEEtran}

\usepackage{amsmath,amssymb,amsthm}    
\usepackage[dvips]{graphicx}   
\usepackage{cite}
\usepackage{bbm}
\usepackage{color}

\newcommand{\nn}{\nonumber}

\newtheorem{theorem}{Theorem}
\newtheorem{lemma}{Lemma}
\newtheorem{proposition}{Proposition}

\newtheorem{property}{Property}

\newtheorem{example}{Example}

\begin{document}
\title{The Non-Regular CEO Problem}
\author{Aditya~Vempaty and Lav~R.~Varshney
\thanks{A.~Vempaty is with the Department of Electrical Engineering and Computer Science, Syracuse University, Syracuse, NY 13244 USA. (e-mail: avempaty@syr.edu).}%
\thanks{L.~R.~Varshney is with the Department of Electrical and Computer Engineering and the Coordinated Science Laboratory, University of Illinois at Urbana-Champaign, Urbana, IL 61801 USA. (e-mail: varshney@illinois.edu).}}
\maketitle

\begin{abstract}
We consider the CEO problem for non-regular source distributions (such as uniform or truncated Gaussian). A group of agents observe independently corrupted versions of data and transmit coded versions over rate-limited links to a CEO. The CEO then estimates the underlying data based on the received coded observations. Agents are not allowed to convene before transmitting their observations. This formulation is motivated by the practical problem of a firm's CEO estimating (non-regular) beliefs about a sequence of events, before acting on them. Agents' observations are modeled as jointly distributed with the underlying data through a given conditional probability density function.  We study the asymptotic behavior of the minimum achievable mean squared error distortion at the CEO in the limit when the number of agents $L$ and the sum rate $R$ tend to infinity. We establish a $1/R^2$ convergence of the distortion, an intermediate regime of performance between the exponential behavior in discrete CEO problems [Berger, Zhang, and Viswanathan (1996)], and the $1/R$ behavior in Gaussian CEO problems  [Viswanathan and Berger (1997)]. Achievability is proved by a layered architecture with scalar quantization, distributed entropy coding, and midrange estimation. The converse is proved using the Bayesian Chazan-Zakai-Ziv bound.
\end{abstract}

\begin{IEEEkeywords}
multiterminal source coding, mean-square error, midrange estimator, Chazan-Zakai-Ziv bound
\end{IEEEkeywords}

\section{Introduction}
Consider the problem faced by the chief executive officer (CEO) of a firm with a large portfolio of projects that each have an underlying probability of success, say drawn from a uniform distribution on $[0,1]$.  Each of her subordinates will have noisy beliefs about the risks facing the projects: random variables jointly distributed with success probability, e.g.\ according to a copula model (common in mathematical finance to model beliefs about risks\cite{Nelsen2006,GenestM1986,CherubiniLV2004}).  Subordinates must convey these risks, but are not allowed to convene.  The CEO has cognitive constraints that limit the information rate she can receive from her subordinates, requiring subordinates to partition risks into quantal grades like A, B, C, and D, before conveyance.  Such quantized grading is typical in businesses with complex information technology projects \cite{RatakondaWBTG2010}.  Upon receiving information from the subordinate agents, the CEO estimates the underlying success probability to minimize mean squared error (Brier score \cite{PreddSLOPK2009}) before taking action.

This is, of course, a version of the classical problem in multiterminal source coding called the 
CEO problem, but for a previously unanalyzed source.  The CEO problem was first introduced by Berger, Zhang, and Viswanathan \cite{BergerZV1996}, where they considered the source of interest to be a discrete data sequence and studied the asymptotic behavior of the minimal error probability in the limit as the number of agents and the sum rate tend to infinity. It was later extended to the case when the source sequence of interest is continuous Gaussian distributed and a quadratic distortion measure is used as a performance measure \cite{ViswanathanB1997}. Oohama studied the sum rate distortion function of the quadratic Gaussian CEO problem and determined the complete solution to the problem \cite{Oohama1998}; the full rate-distortion region was then found independently by Oohama \cite{Oohama2005} and Prabhakaran, et al.~\cite{PrabhakaranTR2004}. Several extensions to this problem have been studied \cite{YangX2012,TavildarV2005,ChenW2011,WagnerTV2008,ChenZBW2004,GastparE2005,SoundararajanWV2012}, however most of these extensions continue to deal with the quadratic Gaussian setting. Viswanath formulated a similar multiterminal Gaussian source coding problem and characterized the sum rate distortion function for a class of quadratic distortion metrics \cite{Viswanath2004}. In \cite{TavildarV2005,ChenW2011}, the authors consider the vector Gaussian case and study the sum rate for vector Gaussian CEO problem. The related problem of determining the rate region of the quadratic Gaussian two-encoder source coding problem was solved by Wagner, Tavildar, and Viswanath \cite{WagnerTV2008}. Chen, et al. determined bounds on the rate region for the CEO problem with general source distributions \cite{ChenZBW2004}. Eswaran and Gastpar considered the CEO problem where the source is non-Gaussian but observations are still made through additive white Gaussian noise (AWGN) channel \cite{GastparE2005}. 

The motivating example of belief sharing in organizations is also closely connected to studies of communicating probability values \cite{Hildreth1963,KramerS2007,VarshneyV2008,RhimVG2012,VarshneyV2014}. Contrary to typical work in distributed source coding for inference \cite{HanA1998} that is concerned with compressing agents' measurements, \cite{VarshneyV2008} studied optimal quantization of prior probabilities for Bayesian hypothesis testing.  This was extended to the case of collaborative decision making in parallel fusion settings \cite{RhimVG2012}, very much like the CEO problem herein but in non-asymptotic regimes. Such a problem arises in several statistical signal processing, economics, and political science settings such as human affairs, where juries or committees need to possess a common preference for two alternatives.

In sensor network settings where sensors see a phenomenon through Gaussian noise but produce censored data due to hardware limitations of the measuring device, observations might follow truncated Gaussian with bounded support. Censored sensor data renders celebrated information-theoretic results for Gaussian observations invalid.

Motivated by such applications, we consider the CEO problem with non-regular source distributions (in the sense of Bayesian estimation theory \cite[p.~72]{VanTrees1968}), and observations governed by a given conditional probability density function, for example, through copula models.  More precisely we consider an i.i.d.\ source sequence $X(t)$, which follows a probability density function $f_X(x)$ with finite support $\mathcal{X}$, such that
\[
\frac{\partial f_X(x)}{\partial x} \mbox{ or } \frac{\partial^2 f_X(x)}{\partial x^2}
\]
either does not exist or is not absolutely integrable. We determine the asymptotic behavior of quadratic distortion as a function of sum rate in the limit of large numbers of agents and rate.

As commented by Viswanathan and Berger \cite{ViswanathanB1997}, results for discrete and continuous alphabets are not very different in most problems of information theory. However, for the CEO problem, the average distortion decays at an exponential rate for the discrete case and decays as $1/R$ for the Gaussian case, where $R$ is the sum rate. We derive an intermediate $1/R^2$ decay rate behavior when the regularity conditions required for the Bayesian Cram\'{e}r-Rao lower bound used in \cite{ViswanathanB1997} do not hold.  That is, we study the behavior of 
$$\beta\triangleq\lim_{L,R\to\infty}R^2D(R,L)$$ where $L$ is the total number of agents, $R$ is the sum rate, and $D(R,L)$ is the minimum achievable quadratic distortion for a fixed number of total agents $L$ and sum rate $R$. Achievability is proved through a layered scheme that follows quantization, entropy coding, and midrange estimation at the CEO. The converse is proved by lower-bounding the distortion using the extended Bayesian Chazan-Zakai-Ziv bound \cite{ChazanZZ1975,VanTreesB2007}. This result holds for the uniform distribution as a special case, implying that the CEO problem for sharing beliefs or probability values, attains $1/R^2$ convergence. 

The remainder of the paper is organized as follows: Sec.~\ref{sec:prob} provides the mathematical problem formulation and the main result of the paper. Sec.~\ref{sec:direct} proves the direct part of the coding theorem and Sec.~\ref{sec:converse} proves the converse part, using a version of the extended Chazan-Zakai-Ziv bound. Sec.~\ref{sec:disc} present some discussion on extensions and implications of results. 

\section{Problem Description and Main Result}
\label{sec:prob}
Consider an i.i.d.\ source sequence of interest $\{X(t)\}_{t=1}^\infty$ drawn from a non-regular probability density function $f_X(x)$ with finite support $\mathcal{X}$. Without loss of generality, let the source be supported on $[0,1]$.\footnote{Note that an extension to a general finite support is straightforward.} Several agents ($L$) make imperfect conditionally independent assessments of $\{X(t)\}_{t=1}^\infty$, to obtain noisy versions $\{Y_i(t)\}_{t=1}^\infty$ for $i=1,\ldots,L$. The relationship between $X(t)$ and $Y_i(t)$ is governed by a conditionally independent probability density function $W_\alpha(y_i|x)$ for all agents, where $\alpha$ is the coupling parameter. This coupling parameter represents the strength of the dependence between the source $X(t)$ and the observations $Y_i(t)$. The agents separately compress their observations. The CEO is interested in estimating $X(t)$ such that the mean squared error (MSE) between the $n$-block source $X^n=[X(1),\ldots,X(n)]$ and its estimate $\hat{X}^n=[\hat{X}(1),\ldots,\hat{X}(n)]$ is minimized.

Let the source code $\mathcal{C}_i^n$ of rate $R_i^n=(1/n)\log{|\mathcal{C}_i^n|}$ represent the coding scheme used by agent $i$ to encode a block of length $n$ of observed data $\{y_i(t)\}_{t=1}^\infty$. The CEO's estimate is given as $\hat{X}^n=\phi_L^n(C_1^n,\ldots,C_L^n)$ where $\phi_L^n:C_1^n\times\cdots\times C_L^n$ is the CEO's mapping. A specific achievability scheme for an example system, Fig.~\ref{fig:CEO} shown in the sequel, illustrates the basic system structure.

We are interested in the tradeoff between the sum rate $R=\sum_{i=1}^LR_i^n$ and the MSE at the CEO, $D^n(X^n,\hat{X}^n)$, defined as: 
\begin{equation}
D^n(X^n,\hat{X}^n)\triangleq\frac{1}{n}\sum_{t=1}^n(X(t)-\hat{X}(t)))^2\mbox{.}
\end{equation}

For a fixed set of codes, the MSE corresponding to the best estimator at the CEO is given by:
$$
D^n(C_1^n,\ldots,C_L^n)\triangleq\min_{\phi_L^n}D^n(X^n,\phi_L^n(C_1^n,\ldots,C_L^n))\mbox{.}
$$
Also, define the following quantities:
\begin{equation}
D^n(L,R)\triangleq\min_{\{C_i^n\}:\sum_{i=1}^LR_i^n\leq R}D^n(C_1^n,\ldots,C_L^n)\mbox{,}
\end{equation}
\begin{equation}
D(L,R)\triangleq\lim_{n\to\infty}D^n(L,R)\mbox{,}
\end{equation}
and
\begin{equation}
D(R)\triangleq\lim_{L\to\infty}D(L,R)\mbox{.}
\end{equation}
To understand the tradeoff between sum rate and distortion, we study the following quantity:
\[
\beta(\alpha) \triangleq \lim_{R\to\infty}R^2D(R)\mbox{.}
\]

Let $X$ be the generic random variable representing the source and $Y_i$ represent the generic random variable representing agent $i$'s observation where $X$ and $Y_i$ are related through the conditional pdf $W_\alpha(y_i|x)$. We focus on observation channels which satisfy the following property, when placed in sequence with a forward test channel with an output auxiliary random variable $U$. 
\begin{property}
\label{prop}
For a given observation channel $W_\alpha(y|x)$ between $X$ and $Y$, there exists a random variable $U$ such that: 
$X$, $Y$, and $U$ form a Markov chain, $X\to Y \to U$, and
the conditional distribution of $U$ given $X$, $f_{U|X}(u|x)$, has bounded support: $u\in[a(x),b(x)]$, where $(a+b)(x)$ is invertible and the inverse function $l(\cdot):=(a+b)^{-1}(\cdot)$ is Lipschitz continuous with Lipschitz constant $K>0$, and further does not vanish at its end points: $\lim_{u\to a(x)\text{ or }b(x)} f_{U|X}(u|x)>0$.
\end{property}

Let the set $\mathcal{S}(W)$ denote the set of random variables $U$ which satisfy the above property for a given observation channel $W$. Explicit examples of channels satisfying this property are provided later in Sec.~\ref{sec:examples}.

We now state the main result of the paper; proofs are developed in the sequel.
\begin{theorem}
When conditional density $W_\alpha$ satisfies Property~\ref{prop}, the following relations hold:
\begin{equation}
\beta(\alpha)\leq\frac{2K^2}{\delta^2}\left(\min_{U\in\mathcal{S}(W)} I(Y;U|X)\right)^2
\end{equation}
and
\begin{equation}
\beta(\alpha)\geq\left(\min_{U: X\to Y\to U}I(Y;U|X)\right)^2\int_{h=0}^\infty h\int_{\theta=0}^1f_X(\theta)e^{-hg(\theta)}d\theta dh
\end{equation}
where $K,\delta>0$ are constants and
\begin{equation}
\label{eq:g}
g(\theta)\triangleq\left\{\frac{d}{d\Delta}-\left[\min_s\log\left(\int W_{\alpha}^s(y|\theta)W_\alpha^{1-s}(y|\theta+\Delta)dy\right)\right]\right\}_{\Delta=0}
\end{equation}
is the first derivative of Chernoff information between the conditional densities $W_\alpha$ of the observation given $x=\theta$ and $x=\theta+\Delta$, evaluated at $\Delta=0$. The minimums are taken over all non-trivial random variables to ensure that the conditional mutual information is non-zero.
\end{theorem}

Notice from the theorem, since $\beta(\alpha)$ is a finite constant, it implies that for a non-regular source distribution, in the limit of large sum rate, the distortion decays as $1/R^2$. This serves as an intermediate regime between the exponential decay of the discrete case \cite{BergerZV1996} and $1/R$ decay of the quadratic Gaussian case\cite{ViswanathanB1997}. This result can be summarized as the fact that sharing beliefs (uniform) is fundamentally easier than sharing measurements (Gaussian), but sharing decisions is even easier (discrete). This shows the effect of the underlying source distribution on the asymptotic estimation performance. When the source has countably finite support set (discrete), we observe an exponential decay. On the other extreme, when the source has an unbounded support set (Gaussian), we observe a $1/R$ decay. Our result is for the source with bounded support (uniform, for example), and we get an intermediate result of $1/R^2$ decay. This suggests the intuitive observation that as the number of possibilities for the source (support) increases, it gets more difficult to communicate the values.

One can also note the similarity in structure of the lower bound of $\beta(\alpha)$ in this problem with other CEO problems \cite{BergerZV1996,ViswanathanB1997}. Most notably, in all cases, there is a minimization of conditional mutual information. Also, the bound here depends on Chernoff information which serves as a divergence metric similar to the Kullback-Leibler divergence from the discrete case \cite{BergerZV1996} and as an information metric similar to Fisher information from the quadratic Gaussian case \cite{ViswanathanB1997}.

\subsection{Examples of Observation Channels Satisfying Property~\ref{prop}}
\label{sec:examples}
Property~\ref{prop} may seem a little opaque, so here we give an illustrative example of a family of observation channels that satisfy it.

\begin{proposition}
A sufficient condition for an observation channel to satisfy Property~\ref{prop} is when its density $W_\alpha(y_i|x)$ is given by a copula conditional density function\footnote{A copula is a multivariate probability distribution for which the marginal probability distribution of each variable is uniform\cite{Nelsen2006}.}  and has discontinuity at end points.
\end{proposition}
\begin{proof}
Let the end points of observation channel $W_\alpha(y_i|x)$ be denoted by $e_l(x)$ and $e_u(x)$. Due to its discontinuity at the end points, we have the following
\begin{equation}
\lim_{y_i\to e_l(x)\text{ or }e_u(x)} W_{\alpha}(y_i|x)>0.
\end{equation} 
Consider the test channel given by $U_i=Y_i+N$, where $N$ is a $k$-peak noise for $k\geq2$, $f_N(n)=\sum_{l=1}^k p_l\delta(n-n_l)$, where $\delta(\cdot)$ is the Dirac delta function, $\sum_{l=1}^kp_l=1$, and $n_1<n_2<\cdots<n_k$. Then the conditional distribution of $U$ given $X$, $f_{U|X}(u_i|x)=\sum_{l=1}^kp_lW_{\alpha}(u_i-n_l|x)$, has bounded support: $[e_l(x)+n_1,e_u(x)+n_k]$ and the values of $f_{U|X}(u_i|x)$ at the end points are given by 
\begin{equation}
\lim_{u_i\to e_l(x)+n_1}f_{U|X}(u_i|x)\geq p_1\lim_{y_i\to e_l(x)}W_{\alpha}(y_i|x)>0
\end{equation}
and
\begin{equation}
\lim_{u_i\to e_u(x)+n_k}f_{U|X}(u_i|x)\geq p_k\lim_{y_i\to e_u(x)}W_{\alpha}(y_i|x)>0.
\end{equation}
This proves the proposition.
\end{proof}
We now provide a specific example from the above family of observation channels and explicitly show that it satisfies Property~\ref{prop}.

\begin{example}
As a specific example, consider the case when the source $X(t)$ and the observations $Y_i(t)$ are marginally distributed with uniform distribution in $(0,1)$ and Clayton copula model is used to model the channel between the source and the observations. The conditional distribution of $Y_i$ given $X$ is given by the following (for $1/2<\alpha<1$)
\begin{equation}
\label{eq:cop}
W_\alpha(y_i|x)=
\begin{cases}
(1-\alpha)(xy_i)^{\alpha-1}\left(x^\alpha+y_i^\alpha-1\right)^{1/\alpha-2}, &\text{for $(1-x^\alpha)^{1/\alpha}\leq y_i\leq1$}\\
0, & \text{otherwise.}
\end{cases}
\end{equation}

Using the test channel defined as $U_i=Y_i+N$, where $N$ is a $k$-peak noise for $k\geq2$, $f_N(n)=\sum_{l=1}^k p_l\delta(n-n_l)$, where $\delta(\cdot)$ is the Dirac delta function and $\sum_{l=1}^kp_l=1$, would result in a density $f_{U|X}(u_i|x)$ given by (without loss of generality, assume $n_1<n_2<\cdots<n_k$):
\begin{eqnarray}
\label{eq:fux}
f_{U|X}(u_i|x)=
\begin{cases}
p_1(1-\alpha)(x(u_i-n_1))^{\alpha-1}\left(x^\alpha+(u_i-n_1)^\alpha-1\right)^{1/\alpha-2}, &\text{for $(1-x^\alpha)^{1/\alpha}+n_1\leq u_i\leq1+n_1$}\\
p_2(1-\alpha)(x(u_i-n_2))^{\alpha-1}\left(x^\alpha+(u_i-n_2)^\alpha-1\right)^{1/\alpha-2}, &\text{for $(1-x^\alpha)^{1/\alpha}+n_2\leq u_i\leq1+n_2$}\\
\qquad\qquad\vdots\\
p_k(1-\alpha)(x(u_i-n_k))^{\alpha-1}\left(x^\alpha+(u_i-n_k)^\alpha-1\right)^{1/\alpha-2}, &\text{for $(1-x^\alpha)^{1/\alpha}+n_k\leq u_i\leq1+n_k$}\\
0, & \text{otherwise.}
\end{cases}
\end{eqnarray}
when $1+n_l<(1-x^\alpha)^{1/\alpha}+n_{l+1}$ for $l=1,\ldots,k-1$ to ensure the shifted versions of $W_\alpha(y_i|x)$ do not overlap.\footnote{It is straightforward to prove that the property holds even when there is overlap among the shifted versions.} 

We now show that $f_{U|X}(u_i|x)$ given by \eqref{eq:fux} satisfies Property~\ref{prop}, which basically consists of two conditions on $f_{U|X}(u|x)$: bounded support and non-vanishing end points. $f_{U|X}(u_i|x)$ given in \eqref{eq:fux} has bounded support as $(1-x^\alpha)^{1/\alpha}+n_1\leq u_i\leq1+n_k$ (irrespective of whether the shifted versions overlap or not). Also, its values at the end points are given by:
\begin{equation}
\lim_{u_i\to(1-x^\alpha)^{1/\alpha}+n_1}f_{U|X}(u_i|x)=p_1(1-\alpha)x^{\alpha-1}(1-x^\alpha)^{1-1/\alpha}(0)^{1/\alpha-2}
\to \infty >0
\end{equation}
as $\alpha>1/2$ and
\begin{equation}
\lim_{u_i\to1+n_k}f_{U|X}(u_i|x)=p_k(1-\alpha)x^{-\alpha}>0
\end{equation}
Hence, it satisfies Property~\ref{prop}.
\end{example}

Note the similarity between the form of this test channel and the random quantizer used by Fix for achieving the rate-distortion function of a uniform source \cite{Fix1978}.

As can be seen later in Sec.~\ref{sec:direct}, the achievability involves use of a test channel, Slepian-Wolf encoding and decoding, and midrange estimation at the CEO. Fig.~\ref{fig:CEO} provides a block diagram outlining these steps for the example considered here.

\begin{figure}[htbp]
\begin{center}
\includegraphics[height=3.25in,width=!]{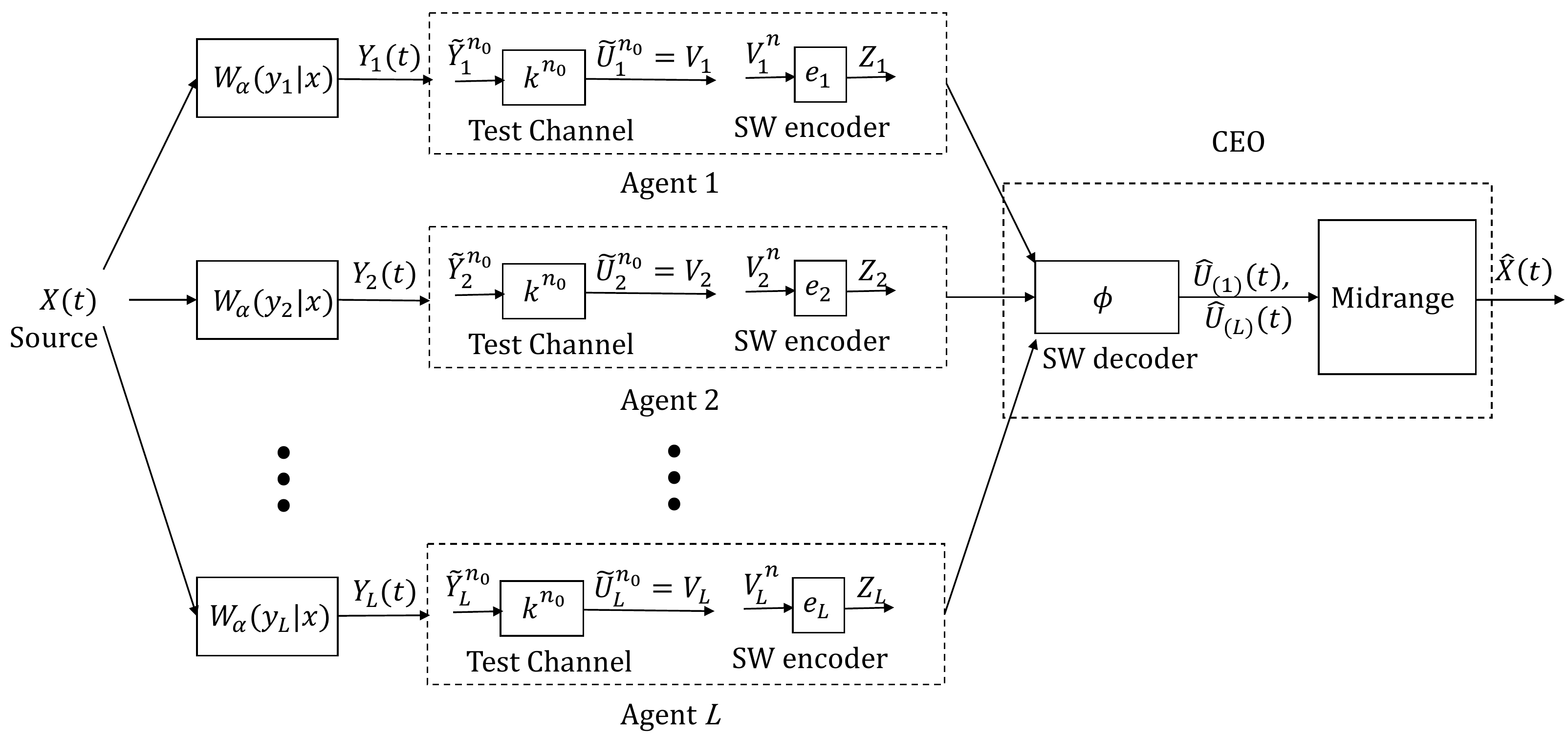}
\caption{A block diagram of system model for the Example~1 using the Clayton copula based observation channel $W_\alpha(y_i|x)$  given by \eqref{eq:cop} and test channel $U=Y+N$ where $N$ is a $k$-peak noise. Here $k^{n_0}$ represents the block code that approximates the test channel.}
\label{fig:CEO}
\end{center}
\end{figure}

\section{Direct Coding Theorem}
\label{sec:direct}
The structure of the achievable scheme is a layered architecture, with scalar quantization followed by Slepian-Wolf entropy coding, just like for the Gaussian CEO problem \cite{ViswanathanB1997} and other source coding problems \cite{ZamirB1999,WagnerTV2008,Servetto2005}. The following are key steps of the analysis: quantization of alphabets, codes that approximate the forward test channel, Slepian-Wolf encoding and decoding, and estimation at the CEO.

Every agent uses a two-stage encoding scheme. In the first stage, a block of observations are mapped to codewords from a codebook which is identical for all agents. The second stage is an index encoder that performs Slepian-Wolf encoding of the codewords \cite{SlepianW1973,Cover1975}. For decoding, the CEO first performs index decoding to determine the $L$ codewords corresponding to each of the agents, and then estimates the source value at each instant based on a midrange estimator \cite{NeymanP1928,ArceF1988}.

The key aspect of the proof is a choice of forward test channel which is characterized by an auxiliary random variable $U$. We choose the test channel from $Y$ to $U$, denoted by $Q(u|y)$ where $U\in\mathcal{S}(W)$, so as to induce a distribution $f_{U|X}(u|x)$ which satisfies Property~\ref{prop}.

\subsection{Quantization of alphabets}
To design the coding scheme, we start by quantizing the continuous alphabets. Denote by $\tilde{X}$, $\tilde{Y}$, and $\tilde{U}$ the quantized versions of random variables $X$, $Y$, and $U$, respectively. Their corresponding alphabets are denoted by $\mathcal{\tilde{X}}$, $\mathcal{\tilde{Y}}$, and $\mathcal{\tilde{U}}$ respectively. Conditions to be satisfied by the quantization are as follows:
\begin{eqnarray}
E(U-\tilde{U})^2\leq{\delta_0}\label{eq:d1}\\
|I(Y;U)-I(\tilde{Y},\tilde{U})|\leq\delta_1\label{eq:d2}\\
|I(X;U)-I(\tilde{X},\tilde{U})|\leq\delta_2\label{eq:d3},
\end{eqnarray}
where $\delta_j>0$, for $j=0,1,2$.
There exists quantization schemes that achieve each of these above constraints individually: \eqref{eq:d1} from that fact that $EU^2<\infty$, and \eqref{eq:d2} and \eqref{eq:d3} from the definition of mutual information for arbitrary ensembles \cite{Dobrushin1959}. Therefore, a common refinement of the quantization schemes that achieve \eqref{eq:d1}--\eqref{eq:d3} separately will satisfy them simultaneously. This quantization induces a corresponding joint probability distribution for the quantized versions $\tilde{X}$, $\tilde{Y}$, and $\tilde{U}$:
\begin{eqnarray*}
P_{\tilde{Y},\tilde{U}}(\tilde{y},\tilde{u})&=&\int_{\{(y,u):\text{quant}(y)=\tilde{y},\text{quant}(u)=\tilde{u}\}}f_{Y,U}(y,u)dydu\\
P_{\tilde{X},\tilde{U}}(\tilde{x},\tilde{u})&=&\int_{\{(x,u):\text{quant}(x)=\tilde{x},\text{quant}(u)=\tilde{u}\}}f_{X}(x)f_{U|X}(u|x)dxdu\nonumber\\
\tilde{W}_\alpha(\tilde{y}|x)&=&\int_{\{y:\text{quant}(y)=\tilde{y}\}}W_\alpha(y|x)dy\\
Q(\tilde{u}|\tilde{y})&=&\frac{P_{\tilde{Y},\tilde{U}}(\tilde{y},\tilde{u})}{P_{\tilde{Y}}(\tilde{y})}.
\end{eqnarray*}
Note that any letters $\tilde{y}$ of zero measure are removed from $\mathcal{\tilde{Y}}$.

\subsection{Codes that approximate the test channel}
The encoding scheme works on the quantized version $\tilde{Y}_i$. The basic idea is to build a block code between the quantized versions $\tilde{Y}_i$ and $\tilde{U}_i$, and show that the designed block code approximates the test channel $Q(u|y)$ that arises from satisfying Property~\ref{prop}. Let $k^{n_0}$ be a block code of length $n_0$ from $\mathcal{\tilde{Y}}^{n_0}$ to $\mathcal{\tilde{U}}^{n_0}$. This map, $k^{n_0}$, induces the following joint distribution between the blocks ${\tilde{Y}}^{n_0}=[\tilde{Y}(1),\ldots,\tilde{Y}(n_0)]$ and ${\tilde{U}}^{n_0}=[\tilde{U}(1),\ldots,\tilde{U}(n_0)]$:
$$\hat{P}^{n_0}({\tilde{Y}}^{n_0}={\tilde{y}}^{n_0},{\tilde{U}}^{n_0}={\tilde{u}}^{n_0})=P_{{\tilde{Y}}^{n_0}}({\tilde{y}}^{n_0})\mathbbm{1}_{\{k^{n_0}({\tilde{y}}^{n_0})={\tilde{u}}^{n_0}\}},$$
where $\mathbbm{1}_{\mathcal{A}}$ is the indicator function which is 1 when the event $\mathcal{A}$ is true and 0 otherwise. Also, the corresponding marginals and conditionals are given by
\[
\hat{P}({\tilde{Y}}(t)=\tilde{y},{\tilde{U}}(t)=\tilde{u})=E_{P_{\tilde{Y}^n}}\mathbbm{1}_{\left\{\tilde{U}(t)=\tilde{u},\tilde{Y}(t)=\tilde{y}\right\}}
\]
\[
\hat{Q}(\tilde{U}(t)=\tilde{u}|\tilde{Y}(t)=\tilde{y})=\frac{\hat{P}(\tilde{Y}(t)=\tilde{y},\tilde{U}(t)=\tilde{u})}{P_{\tilde{Y}}(\tilde{Y}(t)=\tilde{y})}
\mbox{.}
\]

Now the existence of a block code $k^{n_0}:\mathcal{\tilde{Y}}^{n_0}\to\mathcal{\tilde{U}}^{n_0}$ which approximates a test channel $Q(u|y)$ arising from Property~\ref{prop} follows from \cite[Proposition~3.1]{ViswanathanB1997}, which is stated here without proof.
\begin{proposition}[\cite{ViswanathanB1997}]
\label{prop:map}
For every $\epsilon_0,\delta_3>0$, there exists a deterministic map $k^{n_0}:\mathcal{\tilde{{Y}}}^{n_0}\to\tilde{\mathcal{U}}^{n_0}$ with the range cardinality $M$ such that 
\begin{equation}
\label{eq:MtoI}
\frac{1}{n_0}\log{M}\leq I(Y;U)+\delta_3
\end{equation}
and 
$$\sum_{\tilde{u}\in\mathcal{\tilde{U}}}|\hat{Q}(\tilde{U}(t)=\tilde{u}|x)-{Q}(\tilde{U}(t)=\tilde{u}|x)|\leq\frac{\epsilon_0}{|\mathcal{\tilde{X}}|}$$
for all $t=1,\ldots,n_0$ and all real $x$.
\end{proposition}

\subsection{Encoding and decoding}
The encoding is performed in two stages: in the first stage, the agents use the identical deterministic mapping $k^{n_0}$ of Proposition~\ref{prop:map} to encode their quantized observation block $\tilde{Y}_i^{n_0}$ into codewords $\tilde{U}_i^{n_0}$; and in the second stage, Slepian-Wolf encoding \cite{SlepianW1973} is used to encode the index of each agent's codeword $\tilde{U}_i^{n_0}$. Let the index of codeword $\tilde{U}_i^{n_0}$ in the codebook be denoted by $V_i$, for $i=1,\ldots,L$. We will use the index to represent the codeword due to the one-to-one correspondence between the index and the codeword, therefore, we have $V_i=\tilde{U}_i^{n_0}$. Note that $V_1,\ldots,V_L$ are correlated and Slepian-Wolf encoding of the indices is used to remove that correlation across agents. This is done by index encoding $n$-length block of indices of agent $i$, represented as $V_i^n=[V_i(1),\ldots,V_i(n)]$, where $V_i(t)$ is the $t$th component of the $n$-block of the indices of agent $i$. This block of indices is then mapped to a smaller index set using a mapping $e_i:\mathcal{\tilde{U}}^{nn_0}\to\{0,\ldots,N_i-1\}$, for $i=1,\ldots,L$, where $N_i$ and $n$ are chosen to be sufficiently large to ensure a negligible decoding error. The sum rate per source symbol is given by $$R=\frac{1}{nn_0}\sum_{i=1}^L \log N_i.$$ Therefore, we have a complete encoder $h_i=e_i\circ k^{n_0}:\tilde{\mathcal{Y}}^{nn_0}\to\{0,1,2,\ldots,N_i-1\}$, where `$\circ$' is the composition operator. Let the output of this encoder be represented by $Z_i=h_i(\tilde{Y}_i^{nn_0})\in\{0,1,\ldots,N_i-1\}$.

The CEO receives the indices $Z_1,\ldots,Z_L$ corresponding to the $L$ agents. It first recovers the block of indices $\hat{V}_i^n$, for all $i$ using a mapping $\phi:\prod_{i=1}^L\{0,1,\ldots,N_i-1\}\to\prod_{i=1}^L\mathcal{\tilde{U}}_i^{nn_0}$. The output of this decoder, represented as $\hat{V}_i^n=[\hat{V}_i(1),\ldots,\hat{V}_i(n)]=[\hat{U}_i^{n_0}(1),\ldots,\hat{U}_i^{n_0}(n)]$, is the decoded super codeword and $\hat{U}_i^{n_0}(t)$ is the decoded version of $\tilde{U}_i^{n_0}(t)$. From the Slepian-Wolf theorem (cf.~\cite[Proposition 3.2]{ViswanathanB1997}), we know there exist encoders $\{e_i\}$ and a decoder $\phi$ such that the codewords can be recovered with negligible error probability for sufficiently large block size $n$.
\begin{proposition}[\cite{ViswanathanB1997}]
\label{prop:SW}
For every $\epsilon_1,\lambda>0$, there exists sufficiently large $L,n$, and index encoders $e_1,\ldots,e_L$ and index decoder $\phi$ such that
\begin{equation}
\frac{R}{L}\leq\frac{1}{n_0}H(\tilde{U}^{n_0}|\tilde{X}^{n_0})+\epsilon_1,\label{eq:R/L}
\end{equation}
\begin{equation}
\Pr\{(\hat{U}_1^{n_0},\ldots,\hat{U}_L^{n_0})\neq(\tilde{U}_1^{n_0},\ldots,\tilde{U}_L^{n_0})\leq\lambda\}\label{eq:lambda},
\end{equation}
where $\tilde{X}^{n_0}=[\tilde{X}(1),\ldots,\tilde{X}(n_0)]$.
\end{proposition}

\subsection{Further analysis of code rate}
Note that the bound on sum rate per agent $R/L$ in \eqref{eq:R/L} is in terms of the distributions of $\tilde{U}$ and $\tilde{X}$. By further analyzing the code rate, we can determine a bound which is a function of the distributions of the unquantized versions, $X$ and $U$. For this we use the closeness of the marginal distribution induced by the encoding function $k^{n_0}$ to the test channel statistics, to bound the entropy terms. Let $H(\tilde{X})$ denote the entropy of the quantized random variable $\tilde{X}$, then we have
\begin{eqnarray}
\frac{1}{n_0}H(\tilde{U}^{n_0}|\tilde{X}^{n_0})&=&\frac{1}{n_0}(H(\tilde{U}^{n_0},\tilde{X}^{n_0})-H(\tilde{X}^{n_0}))\\
&=&\frac{1}{n_0}H(\tilde{U}^{n_0})+\frac{1}{n_0}H(\tilde{X}^{n_0}|\tilde{U}^{n_0})-\frac{1}{n_0}H(\tilde{X}^{n_0})\\
&=&\frac{1}{n_0}H(\tilde{U}^{n_0})+\frac{1}{n_0}H(\tilde{X}^{n_0}|\tilde{U}^{n_0})-\frac{1}{n_0}\sum_{t=1}^{n_0}H(\tilde{X}(t))\label{eq:indX}\\
&\leq&\frac{1}{n_0}\log{M}+\frac{1}{n_0}H(\tilde{X}^{n_0}|\tilde{U}^{n_0})-\frac{1}{n_0}\sum_{t=1}^{n_0}H(\tilde{X}(t))\label{eq:boundM}\\
&\leq&I(Y;U)+\frac{1}{n_0}H(\tilde{X}^{n_0}|\tilde{U}^{n_0})-\frac{1}{n_0}\sum_{t=1}^{n_0}H(\tilde{X}(t))+\delta_3\label{eq:boundM+I}
\end{eqnarray}
where \eqref{eq:indX} is due to the independent nature of source $X(t)$ over time, \eqref{eq:boundM} is due to the upper bound of $H(\tilde{U}^{n_0})$ by the logarithm of number of codewords $M$, and \eqref{eq:boundM+I} is by \eqref{eq:MtoI}. Next, $H(\tilde{X}^{n_0}|\tilde{U}^{n_0})$ can be further bounded as:
\begin{eqnarray}
H(\tilde{X}^{n_0}|\tilde{U}^{n_0})&=&\sum_{t=1}^{n_0}H(\tilde{X}(t)|\tilde{U}^{n_0},\tilde{X}(1),\ldots,\tilde{X}(t-1))\nonumber\\
&\leq&\sum_{t=1}^{n_0}H(\tilde{X}(t)|\tilde{U}(t))
\end{eqnarray}
using the fact conditioning only reduces entropy. Therefore, we have
\begin{eqnarray}
\frac{1}{n_0}H(\tilde{U}^{n_0}|\tilde{X}^{n_0})&\leq&I(Y;U)+\frac{1}{n_0}H(\tilde{X}^{n_0}|\tilde{U}^{n_0})-\frac{1}{n_0}\sum_{t=1}^{n_0}H(\tilde{X}(t))+\delta_3\\
&\leq&I(Y;U)+\frac{1}{n_0}\sum_{t=1}^{n_0}H(\tilde{X}(t)|\tilde{U}(t))-\frac{1}{n_0}\sum_{t=1}^{n_0}H(\tilde{X}(t))+\delta_3\nonumber\\
&=&I(Y;U)-\frac{1}{n_0}\sum_{t=1}^{n_0}I(\tilde{X}(t),\tilde{U}(t))+\delta_3\\
&\leq&I(Y;U)-I(\tilde{X},\tilde{U})+\delta_3+2\epsilon_0\log \frac{|\mathcal{\tilde{U}}||\mathcal{\tilde{X}}|}{\epsilon_0}\label{eq:symm}
\end{eqnarray}
where \eqref{eq:symm} is due to the $t$-symmetry of the encoder and \cite[Proposition A.3]{ViswanathanB1997}.

Thus, using \eqref{eq:d2}, we have:
\begin{eqnarray}
\frac{1}{n_0}H(\tilde{U}^{n_0}|\tilde{X}^{n_0})&\leq&I(Y;U)-I(\tilde{X},\tilde{U})+\delta_3+2\epsilon_0\log \frac{|\mathcal{\tilde{U}}||\mathcal{\tilde{X}}|}{\epsilon_0}\\
&\leq&I(Y;U)-I(X;U)+\delta_2+\delta_3+2\epsilon_0\log \frac{|\mathcal{\tilde{U}}||\mathcal{\tilde{X}}|}{\epsilon_0}
\end{eqnarray}
due to \eqref{eq:d2}. Due to the Markov chain relationship $X\to Y\to U$, the right hand side can be further simplified to
\begin{eqnarray}
\frac{1}{n_0}H(\tilde{U}^{n_0}|\tilde{X}^{n_0})\leq I(Y;U|X)+\delta_2+\delta_3+2\epsilon_0\log \frac{|\mathcal{\tilde{U}}||\mathcal{\tilde{X}}|}{\epsilon_0}
\end{eqnarray}
By choosing $\delta_2,\delta_3,\epsilon_0,\epsilon_2$ such that
$$\delta_2+\delta_3+2\epsilon_0\log \frac{|\mathcal{\tilde{U}}||\mathcal{\tilde{X}}|}{\epsilon_0}<\epsilon_2,$$
we have 
\begin{eqnarray}
\label{eq:R}
\frac{R}{L}\leq I(Y;U|X)+\epsilon_1+\epsilon_2.
\end{eqnarray}
Having determined a bound on sum rate, the next step is to bound the minimum quadratic distortion.

\subsection{Estimation scheme}
The CEO, after decoding the codewords sent by the agents $(\hat{U}_1^{n_0},\ldots,\hat{U}_L^{n_0})$, estimates the source $X(t)$ on an instant-by-instant basis. Since the range of $U_i(t)$ depends on $X(t)$, we first estimate the midrange of data $\hat{U}_1(t),\ldots,\hat{U}_L(t)$. The midrange estimator \cite{NeymanP1928,Cramer1946,Rider1957,ArceF1988} is the maximally efficient estimator for the center of a uniform distribution. The midrange estimator also seems to work well for estimating the location parameter of other distributions of bounded support and it is more effective than the sample mean for many distributions such as the cosine distribution, parabolic distribution, rectangular distribution, and inverted parabolic distribution \cite{Rider1957}, though the best estimator depends on the distribution of the source that is to be estimated. For these reasons, the midrange estimator is used in this paper.

After estimating the midrange of data, using the inverse function $l(\cdot)$ as follows (cf.~Property~\ref{prop}),
\begin{equation}
\hat{X}(t)=2l\left(\frac{\hat{U}_{(1)}(t)+\hat{U}_{(L)}(t)}{2}\right),
\end{equation}
we get an estimate of $X(t)$. Here $\hat{U}_{(i)}(t)$ are the order statistics of $\hat{U}_{i}(t)$ \cite{DavidN2003}. Note that $E\left[\tfrac{\hat{U}_{(1)}+\hat{U}_{(L)}}{2}\right]=E\left[\tfrac{a(X)+b(X)}{2}\right]$. 

We can now derive an upper bound on the distortion, following a method similar to A\c{c}kay, et al.~\cite{AkcayHL1996}:
\begin{eqnarray}
E[\hat{X}(t)-X(t)]^2&\leq&K^2E\left[\left(\frac{\hat{U}_{(1)}(t)+\hat{U}_{(L)}(t)}{2}-\frac{a(X(t))+b(X(t))}{2}\right)^2\right]\nn\\
&=&K^2E_XE_{U|X}\Bigg[\Bigg(\frac{\hat{U}_{(1)}(t)+\hat{U}_{(L)}(t)}{2}-\frac{a(X(t))+b(X(t))}{2}\Bigg)^2\Bigg|X(t)\Bigg]\\
&\leq&2K^2E_XE_{U|X}\Bigg[\Bigg(\frac{{U}_{(1)}(t)+{U}_{(L)}(t)}{2}-\frac{a(X(t))+b(X(t))}{2}\Bigg)^2\Bigg|X(t)\Bigg]+\epsilon_3\label{eq:D}
\end{eqnarray}
where $U_{(i)}(t)$ are the order statistics of $U_i(t)$; the first inequality is due to Lipschitz continuity of the function $l(\cdot)$ with Lipschitz constant $K$, and \eqref{eq:D} follows from Proposition~\ref{prop:appendix} in the Appendix. 

Now we evaluate the main term in \eqref{eq:D}; for notational simplicity, we drop the dependence on $t$ and the dependence of $a(\cdot)$ and $b(\cdot)$ on $X$. However, we need to be aware of the dependence of the limits $a$ and $b$ on the unknown $X$. 
As $f_{U|X}(u|x)$ does not vanish at the endpoints, there exists $\epsilon$ and $\delta$ such that $f_{U|X}(u|x)\geq\delta$ for $a\leq u\leq a+\epsilon$ and $b-\epsilon\leq u\leq b$. Now,
\begin{align}
E_{U|X}\left[\left(\frac{{U}_{(1)}+{U}_{(L)}}{2}-\frac{a+b}{2}\right)^2\Bigg|X\right]&=\int_{u_{(1)}>a+\epsilon\text{ or }u_{(L)}<b-\epsilon}\left(\frac{{u}_{(1)}+{u}_{(L)}}{2}-\frac{a+b}{2}\right)^2f_{u_{(1)},u_{(L)}|X}du_{(1)}du_{(L)}\nn\\
&+\int_{u_{(1)}<a+\epsilon\text{ and }u_{(L)}>b-\epsilon}\left(\frac{{u}_{(1)}+{u}_{(L)}}{2}-\frac{a+b}{2}\right)^2f_{u_{(1)},u_{(L)}|X}du_{(1)}du_{(L)}\mbox{.}
\label{eq:innervalue}
\end{align}

Since $({u}_{(1)}+{u}_{(L)})/2 \in [a,b]$, we have:
\[
\left(\frac{{u}_{(1)}+{u}_{(L)}}{2}-\frac{a+b}{2}\right)^2\leq\left(\frac{b-a}{2}\right)^2
\]
and the first term on right side of \eqref{eq:innervalue} can be bounded as: 
\begin{align*}
&\int_{u_{(1)}>a+\epsilon\text{ or }u_{(L)}<b-\epsilon}\left(\frac{{u}_{(1)}+{u}_{(L)}}{2}-\frac{a+b}{2}\right)^2f_{u_{(1)},u_{(L)}|X}du_{(1)}du_{(L)} \\ \notag
&\qquad\qquad\leq 2^{-2}(b-a)^2\Pr\left\{u_{(1)}>a+\epsilon\text{ or }u_{(L)}<b-\epsilon|X\right\}\mbox{.}
\end{align*}
Since the $\{U_i\}$ are conditionally independent given $X$, we can simplify this further as:
\begin{align}
\Pr\left\{u_{(1)}>a+\epsilon\text{ or }u_{(L)}<b-\epsilon|X\right\} &\leq \Pr\left\{u_{(1)}>a+\epsilon|X\right\}+ \Pr\left\{u_{(L)}<b-\epsilon|X\right\}\\
&=\prod_{i=1}^L\Pr\left\{u_{i}>a+\epsilon|X\right\}+\prod_{i=1}^L\Pr\left\{u_{i}<b-\epsilon|X\right\}\\
&=\prod_{i=1}^L(1-\Pr\left\{u_{i}\leq a+\epsilon|X\right\})+\prod_{i=1}^L(1-\Pr\left\{u_{i}\geq b-\epsilon|X\right\})\mbox{.}
\end{align}
Since, $f_{U|X}(u|x)\geq\delta$ for $a\leq u\leq a+\epsilon$ and $b-\epsilon\leq u\leq b$, $\Pr\left\{u_{i}\leq a+\epsilon|X\right\}\geq\delta\epsilon$ and $\Pr\left\{u_{i}\geq b-\epsilon|X\right\}\geq\delta\epsilon$. Therefore, 
\begin{align}
\Pr\left\{u_{(1)}>a+\epsilon\text{ or }u_{(L)}<b-\epsilon|X\right\}&\leq \prod_{i=1}^L(1-\Pr\left\{u_{i}\leq a+\epsilon|X\right\})+\prod_{i=1}^L(1-\Pr\left\{u_{i}\geq b-\epsilon|X\right\})\nn\\
&\leq \prod_{i=1}^L(1-\delta\epsilon)+\prod_{i=1}^L(1-\delta\epsilon)\\
&\leq 2(1-\delta\epsilon)^L\mbox{.}
\end{align}

To evaluate the second term in the right side of \eqref{eq:innervalue}, we define the following variables:
\begin{align}
\xi&=L(1-F_{U|X}(u_{(L)})), \ b-\epsilon\leq u_{(L)}\leq b\\
\eta&=LF_{U|X}(u_{(1)}), \ a\leq u_{(1)}\leq a+\epsilon\mbox{,}
\end{align}
where $F_{U|X}$ is the conditional cumulative distribution function of $U$ given $X$.
These variables have the following marginal and joint densities \cite{AkcayHL1996}:
\begin{eqnarray}
&&f_\xi(s)=f_\eta(s)=\left(1-\frac{s}{L}\right)^{L-1},\quad 0\leq s\leq L\\
&&f_{\xi,\eta}(s_1,s_2)=\frac{L-1}{L}\left(1-\frac{s_1+s_2}{L}\right)^{L-2}, \quad s_1,s_2\geq0, \quad s_1+s_2\leq L.\nn
\end{eqnarray}
Also, as $L\to\infty$, $\xi$ and $\eta$ become independent and $f_\xi(s),f_\eta(s)\to e^{-s}$.

From the above definitions, we have 
\begin{eqnarray}
\xi=L\int_{u_{(L)}}^bf_{u|x}du\geq \delta L(b-u_{(L)}),\\
\eta=L\int_{a}^{u_{(1)}}f_{u|x}du\geq \delta L(u_{(1)}-a), 
\end{eqnarray}
provided $u_{(1)}\leq a+\epsilon$ and $b-\epsilon\leq u_{(L)}$. Therefore, for the second term, we have
\begin{eqnarray}
\left|\frac{{u}_{(1)}+{u}_{(L)}}{2}-\frac{a+b}{2}\right|^2&=&\frac{1}{4}\left|(u_{(1)}-a)-(b-{u}_{(L)})\right|^2\\
&\leq&\frac{1}{4}\left[\left|u_{(1)}-a\right|^2+\left|b-{u}_{(L)}\right|^2\right]\\
&\leq&\frac{\xi^2+\eta^2}{4\delta^2L^2},
\end{eqnarray}
where we used the fact that $|A-B|^2\leq A^2+B^2$ for $A,B>0$.

Now using the inequalities we have developed, we can bound the distortion in \eqref{eq:D} as:
\begin{align}
D(L,R)&\leq 2K^2E_X\left[\tfrac{(b-a)^2(1-\delta\epsilon)^L}{2}+\tfrac{1}{4\delta^2L^2}\int_{0}^{L(1-F_{U|X}(b-\epsilon))}\int_{0}^{LF_{U|X}(a+\epsilon)}(s_1^2+s_2^2)f_{\xi,\eta}(s_1,s_2)ds_1ds_2\right]+\epsilon_3\mbox{.}
\label{eq:Dbound}
\end{align} 

Using \eqref{eq:R} and \eqref{eq:Dbound}, we get:
\begin{align*}
R^2D(L,R)&\leq L^2I^2(Y;U|X)\Bigg(2K^2E_X\Bigg[\frac{(b-a)^2(1-\delta\epsilon)^L}{2}\nn\\
&\quad+\frac{1}{4\delta^2L^2}\int_{0}^{L(1-F_{U|X}(b-\epsilon))}\int_{0}^{LF_{U|X}(a+\epsilon)}(s_1^2+s_2^2)f_{\xi,\eta}(s_1,s_2)ds_1ds_2\Bigg]+\epsilon_3\Bigg)\mbox{.}
\end{align*} 
By taking limits $L,R\to\infty$, we have:
\begin{align}
\beta(\alpha)&= \lim_{L,R\to\infty}R^2D(L,R)\nn\\
&\leq I^2(Y;U|X)\Bigg(2K^2E_X\Bigg[\frac{1}{4\delta^2}\int_{0}^{\infty}\int_{0}^{\infty}(s_1^2+s_2^2)\lim_{L\to\infty}f_{\xi,\eta}(s_1,s_2)ds_1ds_2\Bigg]\Bigg)\\
&= I^2(Y;U|X)\Bigg(2K^2E_X\Bigg[\frac{1}{4\delta^2}\int_{0}^{\infty}\int_{0}^{\infty}(s_1^2+s_2^2)e^{-s_1}e^{-s_2}ds_1ds_2\Bigg]\Bigg)\nn\\
&= I^2(Y;U|X)\Bigg(2K^2E_X\Bigg[\frac{1}{2\delta^2}\int_{0}^{\infty}s^2e^{-s}ds\Bigg]\Bigg)\\
&= \frac{2K^2}{\delta^2}I^2(Y;U|X)>0
\end{align} 
where $U$ is chosen to satisfy Property~\ref{prop} and $K>0$ is a constant. Therefore, we have
\begin{equation}
\beta(\alpha)\leq\frac{2K^2}{\delta^2}\left(\min_{U\in\mathcal{S}(W)} I(Y;U|X)\right)^2\mbox{.}
\end{equation}
This concludes the achievability proof. Note that the bound only depends on the conditional mutual information $I(Y;U|X)$ which corresponds to the compression of the observation noise. The compression of the source $X$ does not appear in the bound, since such a term vanishes because the number of agents $L$ grows without bound.

\section{Converse Coding Theorem}
\label{sec:converse}
The converse for the quadratic non-regular CEO problem is similar in structure to the converse for the quadratic Gaussian CEO \cite{ViswanathanB1997} and the discrete CEO problem \cite{BergerZV1996}. The proof uses a lower bound on the distortion function similar to the Bayesian Cram\'{e}r-Rao lower bound used in \cite{ViswanathanB1997}. However, note that the source distribution herein does not satisfy the regularity conditions required for using the Cram\'{e}r-Rao bound \cite{VanTrees1968}. Therefore, we use a version of the extended Chazan-Zakai-Ziv bound \cite{ChazanZZ1975,BellSEV1997,Bell1995} which is first stated here without proof.
\begin{lemma}
\label{lemma:CZZ}
For estimating a random scalar parameter $x\sim f_X(x)$ with support on $[0,T]$ using data $\mathbf{z}=[z_1,\ldots,z_k]$ with conditional distribution $f(\mathbf{z}|x)$, the MSE between $x$ and $\hat{x}(\mathbf{z})$ is bounded as follows:
\begin{equation}
E(x-\hat{x}(\mathbf{z}))^2\geq \frac{1}{2T}\int_{h=0}^Th\left[\int_{\theta=0}^{T-h}(f_X(\theta)+f_X(\theta+h))P_{min}(\theta,\theta+h)d\theta\right]dh,
\end{equation}
where $P_{min}(\theta,\theta+h)$ is the minimum error probability corresponding to the following binary hypothesis testing problem:
\begin{eqnarray}
H_0&:& \mathbf{z}\sim f(\mathbf{z}|x), \quad x=\theta, \qquad \Pr(H_0)=\frac{f_X(\theta)}{f_X(\theta)+f_X(\theta+h)},\nonumber\\
H_1&:& \mathbf{z}\sim f(\mathbf{z}|x), \quad x=\theta+h, \qquad \Pr(H_1)=\frac{f_X(\theta+h)}{f_X(\theta)+f_X(\theta+h)}\nonumber.
\end{eqnarray}
\end{lemma}

The above Chazan-Zakai-Ziv bound falls under the family of Ziv-Zakai bounds. Ziv-Zakai bounds have been shown to be useful bounds for all regions of operation unlike other bounds (for example, Cram\'{e}r-Rao bound) that have limited applicability \cite{Bell1995}. This family of bounds build on the original Ziv-Zakai bound \cite{ZivZ1969} and have the advantage of being independent of bias and very tight in most cases. A detailed study of this family of bounds can be found in \cite{Bell1995}.

Note that this lemma bounds the performance of an estimation problem in terms of the performance of a sequence of detection problems. Therefore, as we shall see, we get Chernoff information rather than Fisher information as seen in the estimation problem in the quadratic Gaussian CEO \cite{ViswanathanB1997}.

Using Lemma~\ref{lemma:CZZ}, we now prove our converse. Let $\{\mathcal{C}_i^n\}_{i=1}^L$ be $L$ codes of block length $n$, corresponding to the $L$ agents, with respective rates $R_1,R_2,\ldots,R_L$. We use the genie-aided approach to determine the lower bound as follows: Let the CEO implement $n$ estimators $O_t$ for $t=1,\ldots,n$ where $O_t$ estimates $X(t)$ given all components of the source word $x^n$ except $x(t)$. Recall the definition of $X^n=[X(1),\ldots,X(n)]$ and further define $Y_i^{n}=[Y_i(1),\ldots,Y_i(n)]$. We have
\begin{align}
nR_i&= \log{|\mathcal{C}_i^n|} \notag \\
&\geq I(Y_i^n;C_i|X^n) \notag \\
&= \sum_{t=1}^nI(Y_i(t);C_i|Y_i^{t-1},X^n)\label{eq:product}\\
&= \sum_{t=1}^n\left[h\left(Y_i(t)|Y_i^{t-1},X^n\right)-h\left(Y_i(t)|C_i,Y_i^{t-1},X^n\right)\right]\nn\\
&= \sum_{t=1}^n\left[h\left(Y_i(t)|X^n\right)-h\left(Y_i(t)|C_i,Y_i^{t-1},X^n\right)\right]\label{eq:ind}\\
&\geq \sum_{t=1}^n[h(Y_i(t)|X^n)-h(Y_i(t)|C_i,X^n)]\label{eq:cond}\\
&=\sum_{t=1}^nI(Y_i(t);C_i|X^n)\mbox{,} \notag
\end{align}
where $X$ is the generic source random variable, $Y_i$ is the noisy version of $X$ as observed by agent $i$; \eqref{eq:product} is from the product rule of mutual information, \eqref{eq:ind} is due to the independence of $Y(t)$ across time, and \eqref{eq:cond} follows since conditioning only reduces entropy.

Hence, we get a lower bound on the sum rate $R$ as follows:
$$R\geq\frac{1}{n}\sum_{t=1}^n\sum_{i=1}^LI(Y_i(t);C_i|X^n).$$

Define $\breve{X}_t=(X_1,\ldots,X_{t-1},X_{t+1},\ldots,X_n)$ and let $U_i(t,\breve{x}_t)$ be a random variable whose joint distribution with $X(t)$ and $Y_i(t)$ is:
\begin{align*}
&\Pr\left\{x\leq X(t)\leq x+dx,y\leq Y_i(t)\leq y+dy,U_i(t,\breve{x}_t)=c\right\}\nn\\
&\quad=f_X(x)W_\alpha(y|x)\Pr(C_i=c|Y_i(t)=y,X(t)=x,\breve{X}_t=\breve{x}_t)dxdy\nn\\
&\quad=f_X(x)W_\alpha(y|x)\Pr(C_i=c|Y_i(t)=y,\breve{X}_t=\breve{x}_t)dxdy \mbox{,}
\end{align*}
since the codeword $C_i$ depends on $X(t)$ only through $Y_i(t)$. Therefore, for each $i$ and any fixed $\breve{x}_t$, we have the Markov chain relationship $X(t)\to Y_i(t) \to U_i(t,\breve{x}_t)$. Now, we can express the lower bound on $R$ as 
\begin{equation}
R\geq\frac{1}{n}\sum_{t=1}^n\sum_{i=1}^LE_{\breve{X}_t}I(Y_i(t);U_i(t,\breve{X}_t)|X(t))\mbox{.}
\label{eq:convR}
\end{equation}

Note that in order to find a lower bound on $\beta(\alpha)$, we consider the best case where the CEO knows $C_1,\ldots,C_L$ and $\breve{x}_t$, i.e., the CEO uses an estimator $\hat{X}(C_1,\ldots,C_L,\breve{x}_t)$. Using the Chazan-Zakai-Ziv bound (Lemma~\ref{lemma:CZZ}), we have:
\begin{equation}
E(X(t)-\hat{X}_t)^2\geq \frac{1}{2} \int_{h=0}^1h\left[\int_{\theta=0}^{1-h}(f_X(\theta)+f_X(\theta+h))P_{min,t}(\theta,\theta+h)d\theta\right]dh
\end{equation}
where $P_{min,t}(\theta,\theta+h)$ is the minimum achievable error probability, using data $Y_1(t),\ldots,Y_L(t)$ from the $L$ agents, to differentiate between $X(t)=\theta$ and $X(t)=\theta+h$.

Therefore, from the definition of $D(L,R)$, we have:
\begin{align}
D(L,R)&= \frac{1}{n}\sum_{t=1}^nE(X(t)-\hat{X}_t)^2\nn\\
&\geq \frac{1}{2n}\sum_{t=1}^n\left[\int_{h=0}^1h\left[\int_{\theta=0}^{1-h}(f_X(\theta)+f_X(\theta+h))P_{min,t}(\theta,\theta+h)d\theta\right]dh\right]\nn\\
&\geq \frac{1}{2nL^2}\sum_{t=1}^n\left[\int_{h=0}^1hL\left[\int_{\theta=0}^{1-h}(f_X(\theta)+f_X(\theta+h))P_{min,t}(\theta,\theta+h)d\theta\right]d(hL)\right]\mbox{,} \nn
\end{align}
where we have multiplied and divided the right side by $L^2$. Now, using a change of variables $\tilde{h}=hL$, we get:
\begin{align}
D(L,R)&\geq \frac{1}{2nL^2}\sum_{t=1}^n\int_{(hL)=0}^L(hL)\int_{\theta=0}^{1-(hL)/L}(f_X(\theta)+f_X(\theta+(hL)/L))P_{min,t}(\theta,\theta+(hL)/L)d\theta d(hL)\nn\\
&= \frac{1}{2nL^2}\sum_{t=1}^n\left[\int_{\tilde{h}=0}^L\tilde{h}\left[\int_{\theta=0}^{1-\tilde{h}/L}(f_X(\theta)+f_X(\theta+\tilde{h}/L))P_{min,t}(\theta,\theta+\tilde{h}/L)d\theta\right]d\tilde{h}\right]\nn\\
&\geq \frac{1}{2L^2}\frac{1}{\frac{1}{n}\sum_{t=1}^n\frac{1}{\left[\int_{\tilde{h}=0}^L\tilde{h}\left[\int_{\theta=0}^{1-\tilde{h}/L}(f_X(\theta)+f_X(\theta+\tilde{h}/L))P_{min,t}(\theta,\theta+\tilde{h}/L)d\theta\right]d\tilde{h}\right]}}\mbox{,}\label{eq:convD}
\end{align}
where the last step is due to the inequality of arithmetic and harmonic means.

Note that, although not explicit, $D(L,R)$ does depend on $R$. This dependence is implicitly visible via $n$ (see \eqref{eq:convR}). Therefore, as can be observed below in \eqref{eq:further}, the product of $R^2$ and $D(L,R)$ results in a positive constant that is independent of $R$ and does not vanish as $L \to \infty$.

Using \eqref{eq:convR} and \eqref{eq:convD}, we have the following expression:
\begin{align}
R^2D(L,R)&\geq \frac{1}{n^2}\frac{n}{2L^2}\frac{\left(\sum_{t=1}^n\sum_{i=1}^LE_{\breve{X}_t}I(Y_i(t);U_i(t,\breve{X}_t)|X(t))\right)^2}{\sum_{t=1}^n\left[\int_{\tilde{h}=0}^L\tilde{h}\left[\int_{\theta=0}^{1-\tilde{h}/L}(f_X(\theta)+f_X(\theta+\tilde{h}/L))P_{min,t}(\theta,\theta+\tilde{h}/L)d\theta\right]d\tilde{h}\right]^{-1}}\nn\\
&= \frac{1}{2nL^2}\frac{\left(\sum_{t=1}^n\sum_{i=1}^LE_{\breve{X}_t}I(Y_i(t);U_i(t,\breve{X}_t)|X(t))\right)^2}{\sum_{t=1}^n\left[\int_{\tilde{h}=0}^L\tilde{h}\left[\int_{\theta=0}^{1-\tilde{h}/L}(f_X(\theta)+f_X(\theta+\tilde{h}/L))P_{min,t}(\theta,\theta+\tilde{h}/L)d\theta\right]d\tilde{h}\right]^{-1}}\nn\\
&= \frac{1}{2nL^2}\frac{\sum_{t=1}^n\sum_{t'=1}^n\sum_{i=1}^L\sum_{i'=1}^LE_{\breve{X}_t}I(Y_i(t);U_i(t,\breve{X}_t)|X(t))E_{\breve{X}_{t'}}I(Y_{i'}(t');U_{i'}(t',\breve{X}_{t'})|X(t'))}{\sum_{t=1}^n\left[\int_{\tilde{h}=0}^L\tilde{h}\left[\int_{\theta=0}^{1-\tilde{h}/L}(f_X(\theta)+f_X(\theta+\tilde{h}/L))P_{min,t}(\theta,\theta+\tilde{h}/L)d\theta\right]d\tilde{h}\right]^{-1}}\nn\\
&\geq \min_t\frac{1}{2nL^2}\frac{\sum_{t'=1}^n\sum_{i=1}^L\sum_{i'=1}^LE_{\breve{X}_t}I(Y_i(t);U_i(t,\breve{X}_t)|X(t))E_{\breve{X}_{t'}}I(Y_{i'}(t');U_{i'}(t',\breve{X}_{t'})|X(t'))}{\left[\int_{\tilde{h}=0}^L\tilde{h}\left[\int_{\theta=0}^{1-\tilde{h}/L}(f_X(\theta)+f_X(\theta+\tilde{h}/L))P_{min,t}(\theta,\theta+\tilde{h}/L)d\theta\right]d\tilde{h}\right]^{-1}}\nn\\
&\geq \min_{t,t',i,i'}E_{\breve{X}_t}I(Y_i(t);U_i(t,\breve{X}_t)|X(t))E_{\breve{X}_{t'}}I(Y_{i'}(t');U_{i'}(t',\breve{X}_{t'})|X(t'))\nn\\
& \qquad\qquad\times\int_{\tilde{h}=0}^L\frac{\tilde{h}}{2}\left[\int_{\theta=0}^{1-\tilde{h}/L}(f_X(\theta)+f_X(\theta+\tilde{h}/L))P_{min,t}(\theta,\theta+\tilde{h}/L)d\theta\right]d\tilde{h}\label{eq:further}
\end{align}
where Proposition~\ref{prop:A6} from the Appendix is used for the last two inequalities. Since the input sequence $X(t)$ is i.i.d.\ over time, the minimum for the `primed' variables and the `unprimed' variables is the same. Therefore, we can further simplify the inequality in \eqref{eq:further} as: 
\[
R^2D(L,R) \geq \left(\min_{t,i}E_{\breve{X}_t}I(Y_i(t);U_i(t,\breve{X}_t)|X(t))\right)^2\int_{\tilde{h}=0}^L\frac{\tilde{h}}{2}\left[\int_{\theta=0}^{1-\tilde{h}/L}(f_X(\theta)+f_X(\theta+\tilde{h}/L))P_{min,t}(\theta,\theta+\tilde{h}/L)d\theta\right]d\tilde{h}\mbox{.}
\]
Further simplification gives the following
\[
R^2D(L,R) \geq \left(\min_{t,i,\breve{X}_t}I(Y_i(t);U_i(t,\breve{X}_t)|X(t))\right)^2\int_{\tilde{h}=0}^L\frac{\tilde{h}}{2}\left[\int_{\theta=0}^{1-\tilde{h}/L}(f_X(\theta)+f_X(\theta+\tilde{h}/L))P_{min,t}(\theta,\theta+\tilde{h}/L)d\theta\right]d\tilde{h}\mbox{.}
\]

Now as $L\to\infty$, using the Chernoff-Stein Lemma \cite{CoverT1991}, the error probability $P_{min,t}(\theta,\theta+\tilde{h}/L)$ is given as $e^{-L\mathtt{C}(\theta,\theta+\tilde{h}/L)}$ where $G_\theta(\tilde{h}/L)\triangleq\mathtt{C}(\theta,\theta+\tilde{h}/L)$ is the Chernoff information between the conditional densities of $y$ given $x=\theta$ and $x=\theta+\tilde{h}/L$. It is given by the following
\[
G_\theta(\tilde{h}/L)=-\min_s\log\left(\int W_\alpha^s(y|\theta)W_\alpha^{1-s}(y|\theta+\tilde{h}/L)dy\right)\mbox{.}
\]

Since the argument of $G_\theta(\tilde{h}/L)$ is close to zero as $L\to\infty$, using the Taylor expansion of $G_\theta(\Delta)$ around zero, we get:
\begin{equation}
G_\theta(\Delta)=G_\theta(0)+\Delta G'_\theta(\Delta)|_{\Delta=0}+O(\Delta^2)\mbox{.}
\end{equation}
Using this expansion, we have: 
\begin{align}
e^{-L\mathtt{C}(\theta,\theta+\tilde{h}/L)}&= e^{-LG_\theta(\tilde{h}/L)}\\
&= e^{-L(G_\theta(0)+\tilde{h}/L G'_\theta(\Delta)|_{\Delta=0}+O(L^{-2}))}\nn\\
&= e^{-\tilde{h} G'_\theta(\Delta)|_{\Delta=0}+O(L^{-1})},
\end{align}
since $G_\theta(0)=\mathtt{C}(\theta,\theta)=0$. Therefore,
\begin{align*}
&\lim_{L\to\infty}R^2D(L,R)\geq\nn\\
&\ \lim_{L\to\infty}\left(\min_{t,i,\breve{X}_t}I(Y_i(t);U_i(t,\breve{X}_t)|X(t))\right)^2\int_{\tilde{h}=0}^L\frac{\tilde{h}}{2}\left[\int_{\theta=0}^{1-\tilde{h}/L}(f_X(\theta)+f_X(\theta+\tilde{h}/L))e^{-\tilde{h} G'_\theta(\Delta)|_{\Delta=0}+O(L^{-1})}d\theta\right]d\tilde{h}\nn\\
\end{align*}
which implies
\[
\beta(\alpha)=\lim_{L,R\to\infty}R^2D(L,R)\geq\left(\min_{U: X\to Y\to U}I(Y;U|X)\right)^2\int_{h=0}^\infty h\int_{\theta=0}^1f_X(\theta)e^{-hg(\theta)}d\theta dh\mbox{,}
\]
where $g(\theta)$ is the first derivative of Chernoff information between the conditional densities ($W_\alpha$) of the observation given $x=\theta$ and $x=\theta+\Delta$, evaluated at $\Delta=0$ and is given by: 
\begin{equation}
g(\theta)=\left\{\frac{d}{d\Delta}-\left[\min_s\log\left(\int W_\alpha^s(y|\theta)W_\alpha^{1-s}(y|\theta+\Delta)dy\right)\right]\right\}_{\Delta=0}\mbox{.}
\end{equation}
This concludes the converse proof.

\section{Discussion}
\label{sec:disc}
We considered the non-regular CEO problem, which addresses the practical case where multiple subordinates send quantal grades of their noisy beliefs to the CEO. When the source distribution does not satisfy the regularity conditions, we get an intermediate regime of performance between the discrete CEO problem \cite{BergerZV1996} and the quadratic Gaussian CEO problem \cite{ViswanathanB1997}. A key observation is the rate of convergence depends on Chernoff information. The result expands the literature on CEO problem from the traditional case of Gaussian source distribution and Gaussian channel noise to non-regular source distributions. While the proofs are similar in structure to the traditional CEO problems, they use different techniques, which can also be applied to other non-Gaussian non-regular multiterminal source coding problems. Our results indicate that one can expect a change in behavior for other multiterminal source coding problems as well, when the source follows non-Gaussian non-regular distribution.

There are a number of interesting future directions for research. In this work, we considered only scaling behavior of quadratic non-regular CEO problem. It is desired to derive precise characterizations for sum rate distortion and for full rate-distortion for this non-regular CEO problem as obtained by Oohama \cite{Oohama1998} and Prabhakaran, et al. \cite{PrabhakaranTR2004}, respectively for the quadratic Gaussian CEO problem. Similar to other CEO problems, we can observe the difference in decay rates for distortion between our result and the centralized case when agents can convene. When agents are allowed to convene, the setup is the single-terminal compression problem whose rate-distortion function under MSE was determined by Fix \cite{Fix1978}. However, it has no simple expression and the optimizing solution has support on finite number of mass points. On the other hand, for absolute error distortion measure, rate-distortion function exists in closed form for uniform source \cite{YaoT1978} and it would be interesting to analyze the uniform CEO problem under the absolute error distortion. Gastpar and Eswaran \cite{GastparE2005} have addressed the CEO problem for non-Gaussian sources, but have considered the additive Gaussian noise channel. An interesting variant is when the source follows a regular distribution with a finite support and the measurement noise is modeled using copula. For example, beta distribution satisfies the regularity conditions and has a finite support. Also, for distributions such as cosine, parabolic, and inverted parabolic, midrange (similar to the one used in this paper) is more efficient than mean \cite{Rider1957}. In such cases, it is interesting to explore if the minimum achievable square distortion would still exhibit a $1/R$ convergence behavior. 

\section*{Acknowledgment}
The authors would like to thank Prof. Vivek K Goyal for discussions on midrange estimators, and Prof. Pramod K. Varshney for his support during this project. We would also like to thank the anonymous reviewers for their valuable comments and suggestions that helped us improve the paper.

\bibliographystyle{IEEEtran}
\bibliography{abrv,conf_abrv,CEO_uniform}

\appendix
\begin{proposition}
\label{prop:appendix}
\begin{eqnarray}
K^2E\left[\left(\frac{\hat{U}_{(1)}+\hat{U}_{(L)}}{2}-\frac{a+b}{2}\right)^2\right]&\leq& 2K^2E\left[\left(\frac{{U}_{(1)}+{U}_{(L)}}{2}-\frac{a+b}{2}\right)^2\right]+\epsilon_3
\end{eqnarray}
where $\epsilon_3=\epsilon_3(\delta_0,n)$ can be made arbitrarily small by making $n$ sufficiently large and $\delta_0$ sufficiently small.
\end{proposition}

\begin{IEEEproof}
Let $\mathcal{B}$ be the event $\{\tilde{U}_i\neq\hat{U}_i, \forall i\}.$ By inequality \eqref{eq:lambda}, we have $\Pr\left\{\mathcal{B}\right\}\leq\lambda$. Now, 
\begin{align*}
&K^2E\Bigg[\left(\frac{\hat{U}_{(1)}+\hat{U}_{(L)}}{2}-\frac{a+b}{2}\right)^2-2\left(\frac{{U}_{(1)}+{U}_{(L)}}{2}-\frac{a+b}{2}\right)^2\Bigg]\nn\\
&= \frac{K^2}{4}E\Bigg[\left(\hat{U}_{(1)}+\hat{U}_{(L)}-(a+b)\right)^2-2\left({U}_{(1)}+{U}_{(L)}-(a+b)\right)^2\Bigg]\nn\\
&= \frac{K^2}{4}E\Bigg[\Big((\hat{U}_{(1)}+\hat{U}_{(L)})-({U}_{(1)}+{U}_{(L)})-\left((a+b)-({U}_{(1)}+{U}_{(L)})\right)\Big)^2-2\left({U}_{(1)}+{U}_{(L)}-(a+b)\right)^2\Bigg]\nn\\
&\leq \frac{K^2}{4}E\Bigg[2\left((\hat{U}_{(1)}+\hat{U}_{(L)})-({U}_{(1)}+{U}_{(L)})\right)^2+2\left((a+b)-({U}_{(1)}+{U}_{(L)})\right)^2-2\left({U}_{(1)}+{U}_{(L)}-(a+b)\right)^2\Bigg]\nonumber\\
&= \frac{K^2}{2}E\left[\left((\hat{U}_{(1)}+\hat{U}_{(L)})-({U}_{(1)}+{U}_{(L)})\right)^2\right]\nn\\
&\leq K^2E\left[\left(({U}_{(1)}+{U}_{(L)})-(\tilde{U}_{(1)}+\tilde{U}_{(L)})\right)^2\right]+K^2E\left[\left((\tilde{U}_{(1)}+\tilde{U}_{(L)})-(\hat{U}_{(1)}+\hat{U}_{(L)})\right)^2\right]\nn\\
&\leq 2K^2E\left[\left({U}_{(1)}-\tilde{U}_{(1)}\right)^2\right]+2K^2E\left[\left({U}_{(L)}-\tilde{U}_{(L)}\right)^2\right]+2K^2E\left[\left(\tilde{U}_{(1)}-\hat{U}_{(1)}\right)^2\right]+2K^2E\left[\left(\tilde{U}_{(L)}-\hat{U}_{(L)}\right)^2\right]\nonumber\mbox{,}
\end{align*}
where $\tilde{U}_{(i)}$ are the order statistics of $\tilde{U}_i$, and the last two inequalities follow from the fact that $E[(A+B)^2]\leq2E[A^2]+2E[B^2].$

Now, choose $\delta_0$ to be sufficiently small to ensure that the ordering of variates $U_i$ is preserved under quantization. Then, $U_{(1)}$ and $\tilde{U}_{(1)}$ correspond to the same agent's data, say the $\ell$th agent. Therefore, 
\[
E\left[\left({U}_{(1)}-\tilde{U}_{(1)}\right)^2\right]=E\left[\left({U}_{\ell}-\tilde{U}_{\ell}\right)^2\right]\leq\delta_0
\] 
by \eqref{eq:d1}. Similarly, $E\left[\left({U}_{(L)}-\tilde{U}_{(L)}\right)^2\right]\leq \delta_0$. Also, define $\tilde{u}_{\text{max}}=\text{max}\{|\tilde{u}|:\tilde{u}\in\mathcal{\tilde{U}}\}$. Now, for $i=\{1,\ldots,L\}$:
\begin{align}
E\left[\left(\tilde{U}_{(i)}-\hat{U}_{(i)}\right)^2\right]&= \sum_{u,u'}\left(\tilde{U}_{(i)}-\hat{U}_{(i)}\right)^2\Pr\left\{\tilde{U}_{(i)}=u,\hat{U}_{(i)}=u'\right\}\nn\\
&= \sum_{u,u'}\left(\tilde{U}_{(i)}-\hat{U}_{(i)}\right)^2\Pr\left\{\tilde{U}_{(i)}=u,\hat{U}_{(i)}=u'\right\}\nn\\
&\leq \sum_{u,u'}4\tilde{u}_{\text{max}}^2\Pr\left\{\tilde{U}_{(i)}=u,\hat{U}_{(i)}=u'\right\}\nn\\
&= 4\tilde{u}_{\text{max}}^2\Pr\left\{\tilde{U}_{(i)}\neq\hat{U}_{(i)}\right\}\nn\\
&\leq 4\tilde{u}_{\text{max}}^2\Pr\left\{\mathcal{B}\right\}\nn\\
&\leq 4\tilde{u}_{\text{max}}^2\lambda\mbox{.}\nn
\end{align}

Therefore, 
\begin{align}
K^2E\Bigg[\Bigg(\frac{\hat{U}_{(1)}+\hat{U}_{(L)}}{2}-\frac{a+b}{2}\Bigg)^2-2\left(\frac{{U}_{(1)}+{U}_{(L)}}{2}-\frac{a+b}{2}\right)^2\Bigg]&\leq 4K^2\delta_0+8K^2\tilde{u}_{\text{max}}^2Pr(\mathcal{B})\nn\\
&\leq 4K^2\delta_0+8K^2\tilde{u}_{\text{max}}^2\lambda\nn.
\end{align}

Now, choosing a sufficiently large $n$ such that 
\[
\lambda<\frac{\epsilon_3-4K^2\delta_0}{8K^2\tilde{u}_{\text{max}}^2}\mbox{,}
\]
yields the desired result.
\end{IEEEproof}

\begin{proposition}
\label{prop:A6}
The following inequality:
\begin{equation}
\frac{\sum_{i=1}^np_iA_i}{\sum_{i=1}^np_iB_i}\geq\min_i\left(\frac{A_i}{B_i}\right)
\end{equation}
holds, if $p_i,A_i,B_i\geq 0$ and not all are $0$.
\end{proposition}
\begin{IEEEproof}
Let $m=\min_i\left(\frac{A_i}{B_i}\right)$. By definition,
\begin{eqnarray}
A_i&\geq& B_im, \mbox{ for all } i=1,\ldots,n\nn\\
\implies p_iA_i &\geq&p_iB_im, \mbox{ for all } i=1,\ldots,n\nn\\
\implies \sum_{i=1}^np_iA_i&\geq&m\sum_{i=1}^np_iB_i\nn\\
\implies\frac{\sum_{i=1}^np_iA_i}{\sum_{i=1}^np_iB_i}&\geq&m=\min_i\left(\frac{A_i}{B_i}\right)\nn.
\end{eqnarray}
\end{IEEEproof}
\end{document}